\documentclass[12pt]{amsart}
\theoremstyle{definition}
\newtheorem{definition}{Definition}
\newtheorem{theorem}{Theorem}[section]

\usepackage{url,upref,amscd,amsmath,amsthm}
\usepackage{times}
\usepackage{fancyhdr}
\usepackage[psamsfonts]{amssymb}
\usepackage{cite}
\title{The discrete logarithm problem in the group of non-singular circulant matrices}
\author{Ayan Mahalanobis}
\address{Indian Institute of Science Education and Research, Pashan, Pune
  411021 India}
\email{ayanm@iiserpune.ac.in}
\begin{document}
\today
\begin{abstract}
The discrete logarithm problem is one of the backbones in public key
cryptography. In this paper we study the discrete logarithm problem in
the group of circulant matrices over a finite field. This gives rise
to secure and fast public key cryptosystems.
\end{abstract}
\maketitle
\section{Introduction}
Menezes and Wu~\cite{menezes} claim that working with the discrete
logarithm problem in matrices offers
no major improvement from working with a finite field. Many authors,
including myself~\cite{ayan1}, repeated that claim. It is now a common
knowledge that for practical purposes, the discrete logarithm problem in
non-singular matrices is not worth looking at. 

In this note, I provide a counterexample to the above mentioned common
knowledge and show that matrices can be used effectively to produce a
fast and secure cryptosystem. This approach can be seen as working
with the MOR cryptosystem~\cite{ayan2}, with finite dimensional vector spaces over a
finite field.

In this note, we will only deal with the discrete logarithm problem
in matrices, i.e., given a non-singular $d\times d$ matrix $A$ and
$B=A^m$ over $\mathbb{F}_q$, compute $m$; where $q$ is a power of a
prime $p$. One can easily build any cryptosystem that uses
the discrete logarithm problem, like the Diffie-Hellman key exchange
or the ElGamal cryptosystem, using the discrete logarithm problem in matrices. There are many
aspects to the security of a cryptosystem. In this paper we will only
deal with the computational aspects of solving a discrete
logarithm problem.

The core of the Menezes-Wu algorithm is to compute the
characteristic polynomial $\chi_A(x)$ of $A$. The eigenvalues of $A$,
which are the roots of $\chi_A(x)$
belong to the splitting field of $\chi_A(x)$. The roots of $\chi_{B}(x)$
also belong to the same splitting field. Then to solve the discrete
logarithm problem, one has to solve the individual discrete logarithm
problems in the eigenvalues and then use the Chinese remainder
theorem. The security of the discrete logarithm problem depends on the
degree of the extension of the splitting field. Since solving a discrete logarithm
problem depends on the size of the field, we can get excellent
security by taking $d$ large (around 20) and choose $A$ such that
$\chi_A(x)$ is irreducible. However, in that case matrix
multiplication becomes very expensive and we are better off working
with the finite field $\mathbb{F}_{q^d}$. This is the argument of
Menezes and Wu~\cite{menezes}.

In this paper, we deal with a particular type of non-singular matrices --
the \emph{circulant matrices}. We show, that for these matrices, squaring is
free and multiplication is easy. When this is the case, the above
argument is no longer valid and we have a good chance of a successful
cryptosystem. Using the extended Euclidean algorithm, computing the
inverse of a circulant matrix is easy, that makes a cryptosystem built
on circulant matrices very fast and secure.

When working with the discrete logarithm problem in matrices, one
should be careful of the fact that the determinant of a matrix is a
multiplicative function to the ground field. This can always reduce the discrete
logarithm problem in matrices to a discrete logarithm problem in the
ground field. This can be easily avoided by:
\begin{itemize}
\item[(i)] Choose $A$ such that determinant of $A$ is $1$.
\end{itemize}
\section{Circulant Matrices}
The reader is reminded that all fields (often denoted by $F$) are
finite with characteristic $p$.
\begin{definition}
A $d\times d$ matrix over a field $F$ is called circulant, if every row except
the first row, is a right circular shift of the row above that. So a
circulant matrix is defined by its first row.  One can define a circulant 
matrix similarly using columns. 
\end{definition}
Even though a circulant matrix is a two dimensional object, in
practice it behaves much like an one dimensional object given by the
first row or the first column. We will denote
a circulant matrix $C$ with first row $c_0,c_1,\ldots,c_{d-1}$ by
$C=\text{circ}\left(c_0,c_1,c_2,\ldots,c_{d-1}\right)$.
An example of a circulant $5\times 5$ matrix is:
\[
\begin{pmatrix}
c_0&c_1&c_2&c_3&c_4\\
c_4&c_0&c_1&c_2&c_3\\
c_3&c_4&c_0&c_1&c_2\\
c_2&c_3&c_4&c_0&c_1\\
c_1&c_2&c_3&c_4&c_0
\end{pmatrix}
\]
It is easy to see that all the (sub)diagonals of a circulant matrix are
constant. This fact comes in handy. Let
$W=\text{circ}(0,1,0,\ldots,0)$ be a $d\times d$ circulant matrix,
then clearly $W^d=I$. We can
write $C=c_0I+c_1W+c_2W^2+\ldots+c_{d-1}W^{d-1}$. One can define a
\emph{representer polynomial} corresponding to the circulant matrix
$C$ as $\phi_C=c_0+c_1x+c_2x^2+\ldots+c_{d-1}x^{d-1}$. This shows
that the circulants form a commutative ring with respect to matrix
multiplication and matrix addition and is isomorphic to (the
isomorphism being matrix to representer polynomial)
$\dfrac{F[x]}{x^d-1}$. For more on circulant matrices, see \cite{davis}.
\subsection{How easy is it to square a circulant matrix?}\hspace{0cm}\\
Let $A=\text{circ}(a_0,a_1,\ldots,a_{d-1})$ be a circulant matrix over
a field of characteristic 2. We
show that to compute $A^2$, we need to compute $a_i^2$
for each $i$ in $\{0,1,2,\ldots,d-1\}$. Then
$A^2=\text{circ}\left(a^2_{\pi(0)},a^2_{\pi(1)}\ldots,a^2_{\pi(d-1)}\right)$,
  where $\pi$ is a permutation of $0,1,2,\ldots,(d-1)$. This was also
  observed by Silverman~\cite[Example 3]{silverman}.  
\begin{theorem}\label{theorem1}
If the characteristic of the field $F$ is 2, and $d$ is an odd
integer, then squaring a $d\times d$ circulant matrix $A$ is the same
as squaring $d$ field elements.
\end{theorem}
\begin{proof}
We use the standard method of matrix multiplication; where one
computes the dot product of the i\textsuperscript{th} row with the
j\textsuperscript{th} column for the element at the
intersection of the i\textsuperscript{th} row and the
j\textsuperscript{th} column of the product matrix. As we saw before
the circulant matrices are closed under multiplication and a circulant
matrix is given by its first row.

Taking these into account, if the circulant is
$A=\text{circ}\left(a_0,a_1,\ldots,a_{d-1}\right)$, we see that the
first element of the first row of the product, is the dot product of
$\left(a_0,a_1,\ldots,a_{d-1}\right)$ with the first column
$\left(a_0,a_{d-1},\ldots,a_1\right)^T$. The first column can be thought of
as the map $a_i\mapsto a_{-i \mod d}$ for $i=0,1,\ldots,(d-1)$.

For each $j$ in $\{0,1,2,\ldots,d-1\}$, the map is given by
$a_i\mapsto a_{j-i\mod d}$. Now notice that if $i\mapsto j-i\mod d$,
then $j-i\mapsto i\mod d$. This proves that there are pairs formed in
the dot product,
which makes it zero when working in characteristic 2.

The only thing that escapes forming pairs, are those $i$, for which $i=j-i \mod
d$. Since $d$ is odd, there is an inverse of $2 \mod d$ and 
an unique solution for $i$.
\end{proof}
It is easy to see from the above proof, that once a $d$ is fixed, one
can easily compute the permutation $\pi$. The
computation of $d$ different powers can be done in parallel.
\section{The discrete logarithm problem in the group of non-singular circulant matrices}
As we saw before, circulant matrices can be represented in two
different ways -- one as a circulant matrix and other as an element of
the ring $\mathcal{R}=\dfrac{F[x]}{x^d-1}$. In the later case, each
element of $\mathcal{R}$ is a polynomial of degree $d-1$ in $F$. The
polynomial multiplication in $\mathcal{R}$ can be done (in parallel) using matrix
multiplication. If matrix multiplication is used to do the
polynomial multiplication, then there is no need to do the
reduction mod $x^d-1$.

These two representations lead to two different kinds of attack to the
discrete logarithm problem:
\begin{itemize}
\item[(a)] The discrete logarithm problem in matrices.
\item[(b)] The discrete logarithm problem in $\mathcal{R}$.
\end{itemize}
\subsection{The discrete logarithm problem in matrices}
As we understood from Menezes and Wu~\cite{menezes}, solving the discrete logarithm problem in
non-singular matrices is tied to the largest degree of the irreducible
component of the characteristic polynomial. The best case scenario
happens when the characteristic polynomial is irreducible. For
circulant matrices this is not the case.

It is easy to see that the \emph{row-sum}, sum of all the elements in a
row, is constant in a circulant matrix. This makes the row-sum an
eigenvalue of the matrix. Since this eigenvalue belongs to the ground
field, the only way to escape a discrete logarithm problem in the ground
field is to make sure that the eigenvalue, i.e., the row-sum, is $1$. So the circulant matrix $A$ should be chosen with the following
properties:
\begin{itemize}
\item[(ii)] The matrix $A$ has row-sum $1$.
\item[(iii)] The polynomial $\dfrac{\chi_A}{x-1}$ is
  irreducible.
\end{itemize}
 In the above case the security of the
discrete logarithm problem in $A$ is similar to that of the discrete
logarithm problem in the finite field $\mathbb{F}_{q^{d-1}}$.
\subsection{The discrete logarithm problem in
  $\dfrac{\mathbb{F}_q[x]}{x^d-1}$} 
Notice that
\[\dfrac{\mathbb{F}_q[x]}{x^d-1}\cong\dfrac{\mathbb{F}_q[x]}{x-1}\times\dfrac{\mathbb{F}_q[x]}{\psi(x)},\] 
where $\psi(x)=\dfrac{x^d-1}{x-1}$ and
$\gcd(d,q)=1$. So the discrete logarithm problem in
$\dfrac{\mathbb{F}_q[x]}{x^d-1}$ reduces into two different discrete
logarithm problems, one in the field $\mathbb{F}_q$ and the other in
the ring $\dfrac{\mathbb{F}_q[x]}{\psi(x)}$. The matrix $A$ can be
chosen in such a way that the representer polynomial $\phi_A(x)\mod
(x-1)$ is either $0$ or $1$ and hence reveals no
information about the secret key $m$. If $\psi(x)$ is irreducible, then the
discrete logarithm problem is a discrete logarithm problem in the field
$\dfrac{\mathbb{F}_q[x]}{\psi (x)}$. Hence the security of the
discrete logarithm problem is the same as that of the discrete
logarithm problem in $\mathbb{F}_{q^{d-1}}$.

The question remains, when is $\psi(x)$ irreducible? We know
that~\cite[Theorem 2.45]{lidl},
$x^d-1=\prod\limits_{d_1|d}\Phi_{d_1}(x)$, where $\Phi_k(x)$ is the
$k$\textsuperscript{th} cyclotomic polynomial. It follows
that if $d$ is prime, then $\psi(x)=\Phi_d(x)$. Then the question
reduces to, when is the $d$\textsuperscript{th} cyclotomic polynomial
irreducible, for a prime $d$? It is known~\cite[Theorem 2.47]{lidl} that the
$d$\textsuperscript{th} cyclotomic polynomial $\Phi_d(x)$ is
irreducible over $\mathbb{F}_q$ if and only if $q$ is primitive mod
$d$.

We summarize the requirements on $A$, such that the discrete logarithm
problem is as secure as the discrete logarithm problem in
$\mathbb{F}_{q^{d-1}}$.
\begin{itemize}
\item[(iv)] The integer $d$ is prime.
\item[(v)] The representer polynomial $\phi_A(x)\mod (x-1)$ is
  either $0$ or $1$.
\item[(vi)] $q$ is primitive mod $d$.
\end{itemize}
\section{Why use the discrete logarithm problem with $d\times d$ circulant matrices over $\mathbb{F}_q$ instead of $\mathbb{F}_{q^d}$?} 
A quick answer to the above question is that multiplication in
$\mathcal{R}$, which is isomorphic as algebra to $d\times d$ circulant matrices over $\mathbb{F}_q$, can be much faster!

In implementing the exponentiation in any group, the best known method
is the famous \emph{square-and-multiply} algorithm. Using \emph{normal
  basis}~\cite[Definition 2.32]{lidl},
in a finite field of characteristic 2, squaring is cheap; it is just a cyclic shift
of the bits. In our case, using Theorem~\ref{theorem1}, it is not a
cyclic shift but a permutation. How about multiplication?

The details of the complexity of multiplication is bit
involved, but well studied. So we can skip the details here,  and
refer the reader to \cite{menezes_book,silverman}. The best case
complexity for multiplication in a finite field, using normal basis, is using an \emph{optimal normal
  basis}~\cite[Chapter 5]{menezes_book}. In that case, the complexity
of multiplication in the field
$\mathbb{F}_{2^d}$ is $2d-1$~\cite[Theorem 5.1]{menezes_book}. In the case
of $\mathcal{R}$, that complexity  reduces to
$d$~\cite[Example 3]{silverman}. In $\mathcal{R}$
we get security of $\mathbb{F}_{2^{d-1}}$. So there is an obvious
advantage of working with circulant matrices than with finite fields
-- the complexity of computing the exponentiation reduces to almost
half with only one extra bit. 

Lastly, one can use the extended Euclidean algorithm to compute the inverse of a
representer polynomial in $\mathcal{R}$. In an ElGamal like
cryptosystem, one needs to compute that inverse. This will make
decryption fast.  
\section{Conclusions} 
In this paper we study a discrete logarithm problem in the ring of
circulant matrices. If the matrices are of size $d$, then we saw that
under suitable conditions, 
the discrete logarithm problem is as secure as the discrete logarithm
problem in $\mathbb{F}_{q^{d-1}}$. Since multiplying circulant matrices
is easier, the discrete logarithm problem in
circulant matrix is obviously better than the discrete logarithm
problem in a finite field. 

There is not much history of looking at matrices for better
(more secure) discrete
logarithm problem. In this note the
isomorphism of the circulant matrices with the algebra $\mathcal{R}$
has reduced the central issue of this work to that of implementation of finite
fields. One way to look at $\mathcal{R}$, and this study of the
discrete logarithm problem in $\mathcal{R}$; the finite
field $\mathbb{F}_{q^{d-1}}$ is embedded in $\mathcal{R}$. Though this
is a valid way of looking at the present situation, it is not the
whole view. For example, the issue with row-sum won't be transparent,
unless one chooses to look at matrices. Also this opens up the
possibility that there can be other matrices, in which we can do much
better with the discrete logarithm problem.    
\nocite{*}
\bibliography{paper}
\bibliographystyle{amsplain}
\end{document}